\documentclass[12pt]{amsart}

\usepackage[T1]{fontenc}
\usepackage{xcolor}

\usepackage{braket}
\usepackage{graphicx} 
\usepackage[margin=1in]{geometry}
\usepackage{mathrsfs}
\usepackage{latexsym}
\usepackage{amsmath,amssymb,amsfonts}
\usepackage{cite,hyperref} 
\usepackage[numeric,lite,initials]{amsrefs}
\usepackage{cleveref}
\hypersetup{colorlinks=true,linkcolor=blue}
\usepackage{multirow} 

\usepackage{graphicx}
\usepackage{wrapfig}
\usepackage{fancybox}
\usepackage{bm}
\usepackage{mathtools}
\usepackage{amsmath}
\usepackage{mathdots}
\usepackage{amsthm}
\usepackage{tikz-cd}
\usepackage{thmtools}
\declaretheorem[style=plain, numberwithin=section]{theorem}
\declaretheorem[style=definition,name=Definition,qed=$\blacksquare$, numberwithin=section, sibling=theorem]{definition}
\declaretheorem[style=definition,name=Example,qed={$ \diamondsuit$}, numberwithin=section, sibling=theorem]{example}

\newtheorem{proposition}[theorem]{Proposition}

\newtheorem{remark}[example]{Remark}

\usepackage{url}

\renewcommand{\epsilon}{\varepsilon}%

\newcommand{\SL}{\operatorname{SL}}
\newcommand{\GL}{\operatorname{GL}}
\newcommand{\SU}{\operatorname{SU}}

\newcommand{\U}{\operatorname{U}}
\newcommand{\SO}{\operatorname{SO}}

\newcommand{\CC}{\mathbb C}

\newcommand{\RR}{\mathbb R}
\newcommand{\ZZ}{\mathbb Z}
\renewcommand{\SS}{\mathbb{S}}

\usepackage{color}

\newcommand{\im}{\mathrm i}

\begin{document}
\title{Four-qubit critical states}

\author{Luke Oeding, Ian Tan}
\thanks{Department of Mathematics and Statistics, Auburn University, Auburn, AL (\{\href{mailto:oeding@auburn.edu}{oeding}, \href{mailto:yzt0060@auburn.edu}{yzt0060}\}@auburn.edu).}

\begin{abstract} Verstraete, Dehaene, and De Moor (2003) showed that SLOCC invariants provide entanglement monotones. We observe that many highly entangled or useful four-qubit states that appear in prior literature are stationary points of such entanglement measures. This motivates the search for more stationary points. We use the notion of critical points (in the sense of the Kempf-Ness theorem) together with Vinberg theory to reduce the complexity of the problem significantly. We solve the corresponding systems utilizing modern numerical nonlinear algebra methods and reduce the solutions by natural symmetries. This method produces an extended list of four-qubit stationary points, which includes all the critical states in the survey by Enr\'iquez et al (2016). To illustrate the potential for application, we discuss the use of these states to generate pure five-qubit and six-qubit quantum error correcting codes by reversing a construction of Rains (1996).
\end{abstract}

\maketitle

\section{Introduction}
Let $\mathcal{H}_n:=(\CC^2)^{\otimes n}$ denote the Hilbert space of unnormalized $n$-qubit pure state vectors. The group $\GL_2(\CC)^{\times n}$ has a natural representation $\rho:\GL_2(\CC)^{\times n}\to\GL(\mathcal{H}_n)$ defined by
\begin{equation}\label{eq:rep}
\rho(g_1,\dots,g_n)v_1\otimes\dots\otimes v_n=g_1v_1\otimes\dots\otimes g_n v_n
\end{equation}
where $g_i\in\GL_2(\CC)$ and $v_i\in \CC^2$.
Two important subgroups of $\GL_2(\CC)^{\times n}$ are the \textit{local unitary group} $\U_2^{\times n}$ and the \textit{SLOCC group} $\SL_2(\CC)^{\times n}$ (the group of stochastic local operations with classical communication). Both act on $\mathcal{H}_n$ by restricting $\rho$. Furthermore, the symmetric group $\mathfrak{S}_n$ acts on $\mathcal{H}_n$ by permuting tensor factors:
\[
\sigma.v_1\otimes\dots\otimes v_n = v_{\sigma(1)}\otimes\dots\otimes v_{\sigma(n)},\quad \sigma\in\mathfrak{S}_n.
\]
A polynomial $f:\mathcal{H}_n\to\CC$ is \textit{symmetric} if it is $\mathfrak{S}_n$-invariant.

In quantum information theory one studies the information-theoretical tasks that can be accomplished by utilizing the quantum phenomenon of entanglement. Naturally, one might ask for ways to quantify entanglement. Notions of measurements of entanglement vary in the literature \cite{Nielsen_Chuang}. In this paper, we are interested in entanglement measures of the form
\begin{equation*}
E:\mathcal{H}_n\to[0,\infty),\quad E(\varphi)=|f(\varphi)|^{1/m}
\end{equation*}
where $f\in\CC[\mathcal{H}_n]$ is a homogeneous symmetric SLOCC invariant polynomial of degree $m>0$.
As noted by Osterloh and Siewert \cite{Osterloh_2010}, such a function has the following properties that are desirable for an entanglement measure:
\begin{itemize}
    \item[1.] $E(\varphi)=0$ whenever $\varphi\in\mathcal{H}_n$ is separable (i.e. is a rank-one tensor),
    \item[2.] $E$ is invariant under the action of the local unitary group,
    \item[3.] $E$ is invariant under the action of the symmetric group, and
    \item[4.] $E$ is an entanglement monotone \cite{vidalMonotone}, i.e. it does not increase, on average, under local operations.
\end{itemize}
Let us briefly verify this claim. Since $f$ vanishes on rank-one tensors (this is true for any nonconstant SLOCC invariant), so does $E$. Note that any $A\in \U_2$ has the form $e^{it}A'$ for some $t\in\RR$ and $A'\in\SU_2$. Then the second and third items hold because $E$ is invariant under the SLOCC group action, multiplication by unit complex numbers, and the symmetric group action. The fourth item holds by \cite[Theorem~2]{vddm03}. Examples of entanglement measures with this form include the Christensen-Wong $n$-tangle \cite{WongNelson01}, and those arising from the invariant-comb approach of Osterloh and Siewert \cite{Osterloh_2010}.

In past work \cites{Borras_2007,HIGUCHI2000213,GourWallach2010,OsterlohHD,jaffali2023maximally}, researchers have been interested in finding highly entangled states by maximizing an entanglement measure $E$ over the real manifold $M\subset \mathcal{H}_n$ of unit-norm vectors. If $\varphi\in M$ is such a maximizer, then $\varphi$ is a stationary point of $E|_M$ in the sense that the gradient vanishes. In this paper, we argue that it is worthwhile to understand the complete set of stationary points. We focus on the four-qubit case $n=4$. An important observation is that many known highly entangled states turn out to be stationary points of $E|_M$ for natural choices of $f$.

\subsection{A selection of 4-qubit critical states}\label{sec:selection} We now present some interesting highly-entangled states $\varphi\in\mathcal{H}_4$. All of these states are equivalent to ones that appear in the survey of Enr\'{i}quez et al.~\cite{Enriquez_2016}; see also \cite{e17075063}. Note that these states have the property that every single-qubit reduced density matrix of $\varphi$ maximizes the von Neumann entropy. Such states are called ``critical'' (we define this term in \Cref{sec:critical}). We will write the basis vectors of $\mathcal{H}_n$ as ``ket'' vectors, so that
$
\ket{i_1 i_2 \dots i_n} := e_{i_1}\otimes e_{i_2}\otimes\dots\otimes e_{i_n}
$
where each $i_j\in\{0,1\}$ and $\{e_0,e_1\}$ is the basis for $\CC^2$.

The first state of interest is the well-known GHZ state
\[\ket{GHZ}=\frac{1}{\sqrt{2}}(\ket{0000}+\ket{1111}).\]
Next, the state
\[
\ket{MP}=\frac{1}{2}(\ket{0000}+\ket{0101}+\ket{1010}+\ket{1111})
\]
can be used to win pseudo-telepathy games such as the Mermin-Peres magic square game \cite{KRH2024}. The Yeo-Chua state
\[
\ket{YC}=\frac{1}{\sqrt{8}}(\ket{0000}-\ket{0011}-\ket{0101}+\ket{0110}+\ket{1001}+\ket{1010}+\ket{1100}+\ket{1111}),
\]
can be used to teleport an arbitrary 2-qubit entangled state \cite{YeoChua}. The Higuchi-Sudberry state
\[
\ket{HS}=\frac{1}{\sqrt{6}}[\ket{0011}+\ket{1100}+\omega(\ket{0101}+\ket{1010})+\omega^2(\ket{0110}+\ket{1001})],
\]
where $\omega=e^{2\pi\im/3}$, was found to maximize the average von Neumann entropy across all splittings into bipartite subsystems \cite{HIGUCHI2000213}.
The state
\[
\ket{BSSB} = \frac{1}{\sqrt{12}}[\ket{0110}+\ket{1011}+\im (\ket{0010}+\ket{1111})+(1+\im)(\ket{0101}+\ket{1000})],
\]
found via numerical search by Brown et al., is highly entangled considering the negative eigenvalues of all partial transposes \cite{Brown_2005}.
The cluster states
\begin{align*}
    \ket{C_1}&=\frac{1}{2}(\ket{0000}+\ket{0011}+\ket{1100}-\ket{1111}), \\
    \ket{C_2}&=\frac{1}{2}(\ket{0000}+\ket{0110}+\ket{1001}-\ket{1111}), \\
    \ket{C_3}&=\frac{1}{2}(\ket{0000}+\ket{0101}+\ket{1010}-\ket{1111}),
\end{align*}
can be used in one-way quantum computing; this has in fact been experimentally realized \cite{One-way}. The cluster states were also shown to be the only states that maximize the R\'enyi $\alpha$-entropy for all $\alpha\geq 2$ \cite{GourWallach2010}. Note that $\ket{C_1}$ and $\ket{C_2}$ are related by the transposition $(1,3)\in\mathfrak{S}_4$ swapping the first and third qubits, while $\ket{C_1}$ and $\ket{C_3}$ are related by the transposition $(2,3)\in\mathfrak{S}_4$ swapping the second and third qubits. Hence, the cluster states have the same value on $|f|$ for any symmetric $f:\mathcal{H}_4\to\RR$. The hyperdeterminant state
\[
\ket{HD}=\frac{1}{\sqrt{6}}(\ket{0001}+\ket{0010}+\ket{0100}+\ket{1000}+\sqrt{2}\ket{1111}),
\]
was claimed by Osterloh and Siewert \cite{OsterlohHD} to maximize the 4-qubit hyperdeterminant. Also found in \cite{OsterlohHD} is the state
\[
\ket{OS}=\frac{1}{2}(\ket{0001}+ \ket{0010}+\ket{1100}+\ket{1111})
\]
which has nonzero expectation for the so-called four-qubit sixth-order filter. Finally,
Gour and Wallach \cite{GourWallach2010} show that the state
\begin{multline*}
\ket{L}=\frac{1}{4}[(1+\omega)(\ket{0000}+\ket{1111})+(1-\omega)(\ket{0011}+\ket{1100}) \\ +\omega^2(\ket{0101}+\ket{0110}+\ket{1001}+\ket{1010})]
\end{multline*}
maximizes the $\alpha$-Tsallis entropy for all $\alpha>2$, while the state
\[
\ket{M}=\frac{1}{\sqrt{18}}[(1+\sqrt{3}\im)(\ket{0000}+\ket{1111})+(-1+\sqrt{3}\im)(\ket{0011}+\ket{1100})+\ket{0101}+\ket{1010}]
\]
maximizes the $\alpha$-Tsallis entropy for all $0<\alpha<2$. In \Cref{sec:nf} we introduce an algorithm for computing normal forms of critical states. Applying this algorithm, we find that, up to the joint action of the local unitary and symmetric groups, $\ket{YC}=\ket{C_1}=\ket{OS}$, $\ket{HS}=\ket{M}$, and $\ket{HD}=\ket{L}$.

\begin{table}[b]
    \centering
    \begin{tabular}{c|c|c|c|c|c|c}
        & $\ket{GHZ}$ & $\ket{MP}$ & $\ket{C_1}$  & $\ket{HS}$ & $\ket{HD}$ & $\ket{BSSB}$ \\
        \hline
        $|\mathcal{F}_1|$ & 6 & 6 & 0 & 0 & 0 & 0\\
        $|\mathcal{F}_3|$ & 9 & 6 & 0 & $2.\overline{6}$ &  $2.\overline{6}$ & 0\\
        $|\mathcal{F}_4|$ & 16.5 & 6 & 2.5 & 0 & 0 & 1.875\\
        $|\mathcal{F}_6|$ & 64.125 & 6 & 1.875 & $\sim$0.790 & $\sim$3.01 & 0 \\
        $|\text{Hdet}|$ & 0 & 0 & 0 & 0 & $\sim$5.08E{-5} & $\sim$1.53E-5
    \end{tabular}
    \medskip
    \caption{Absolute values of fundamental invariants and the hyperdeterminant at four-qubit stationary points.}\label{table}
\end{table}

The algebra of symmetric SLOCC invariant polynomials $f\in\CC[\mathcal{H}_4]$ is known to be generated by invariants $\mathcal{F}_1$, $\mathcal{F}_3$, $\mathcal{F}_4$, and $\mathcal{F}_6$ of degrees 2, 6, 8, and 12, respectively. The generating set is not unique; we choose generators from \cite{GourWallach2014}, which are written out explicitly in \Cref{sec:sym}. Let $M$ be the set of unit-norm vectors in $\mathcal{H}_4$. Remarkably, we find that for each of the states $\varphi\in M$ mentioned above, either $f(\varphi)=0$ or $\varphi$ is a nonvanishing stationary point of $|f|_M|^{1/m}$ where $f=\mathcal{F}_k$ for some $k=\{1,3,4,6\}$ or $f=\text{HDet}$ is the hyperdeterminant. \Cref{table} lists the values attained on $|f|$ for each distinct state up to the joint action of the local unitary group and the symmetric group $\mathfrak{S}_4$. A tilde in front of a number indicates that the number is rounded to three significant figures.

\subsection{Overview and organization} In this paper, we compute stationary points of the entanglement measures $|\mathcal{F}_1|$, $|\mathcal{F}_3|$, and $|\mathcal{F}_4|$. The computation relies on theoretical results related to Vinberg theory and the Kempf-Ness theorem. Six of the stationary points from our computation can be used to construct six-qubit absolutely maximally entangled states which give rise to families of pure codes; this illustrates the special entanglement properties of such stationary points.

In \Cref{sec:statcrit} we define the notions of stationary and critical points. An important fact is that, in our setting, nonvanishing stationary points are critical points. In \Cref{sec:vinberg} we bring in results from Vinberg's theory of $\theta$-groups. We define the Cartan subspace and Weyl group for $\mathcal{H}_4$, give a simple algorithm that finds equivalent representatives in the Cartan subspace for critical points, then give expressions for the invariants and their restrictions to the Cartan subspace. In \Cref{sec:compute} we compute nonvanishing stationary points. First we explain how to set up of systems of polynomial equations for the task. In the easiest case, we can describe, with a mathematical proof, all the stationary points of $|\mathcal{F}_1|$. We then find stationary points of $|\mathcal{F}_3|$ and $|\mathcal{F}_4|$ using techniques from numerical algebraic geometry. The states in \Cref{table} are among the solutions found. In \Cref{sec:code} we show an application of stationary points to the construction of absolutely maximally entangled states and pure quantum error correcting codes. We end in \Cref{sec:conclusions} with a discussion and conclusions.

\subsection{Notation and conventions} Let $V$ be an inner product space over a field $\mathbb{K}$. We denote the inner product with angle brackets:
\[
V\times V\to \RR,\quad (v,w)\mapsto \langle v,w\rangle.
\]
If $\mathbb{K}=\CC$, then $\langle v,w\rangle$ is the standard Hermitian inner product. If $\mathbb{K}=\RR$, then $\langle v,w\rangle$ is the standard Euclidean inner product. Since the Hermitian inner product restricts to the Euclidean inner product on real vectors, there is no conflict of notation. The inner product endows $V$ with a norm $\|v\|=\sqrt{\langle v,v\rangle}$.

We shall frequently identify a complex vector space $\CC^N$ with the real vector space $\RR^{2N}$ via
\begin{equation}\label{eq:embed}
\varphi=(z_1,\dots,z_N)\in \CC^N\mapsto \varphi'=(x_1,y_1,\dots,x_N,y_N)\in\RR^{2N}
\end{equation}
where $z_i=x_i+\im y_i$ for each $1\leq i\leq N$. We denote by $\SS^{k-1}$ the real smooth manifold consisting of unit-norm vectors in a $k$-dimensional real vector space. The norm on $\RR^{2N}$ induced by the map $\varphi\mapsto\varphi'$ agrees with the usual norm. Hence we understand that $\SS^{2N-1}\subset(\CC^d)^{\otimes n}$, where $N=d^n$, is the set of unit-norm vectors in $(\CC^d)^{\otimes n}$ and that this is equivalent to the set of unit-norm vectors in $\RR^{2N}$ via the identifications (as real normed spaces) $\mathcal{H}_n=\CC^N=\RR^{2N}$.

We write $\GL_d(\mathbb{K})$, $\SL_d(\mathbb{K})$, and $\SO_d(\mathbb{K})$ for the general linear group, special linear group, and special orthogonal linear group respectively, each consisting of linear maps $\mathbb{K}^d\to\mathbb{K}^d$. It is assumed that $\mathbb{K}=\CC$ when the field is not specified, so that, for example, $\GL_d=\GL_d(\CC)$. Subgroups of $\GL_d^{\times n}$ act on the space of tensors $(\CC^d)^{\otimes n}$ by restricting the obvious generalization of the representation $\rho:\GL_2(\CC)^{\times n}\to\GL(\mathcal{H}_n)$ given in equation~\eqref{eq:rep}.

\section{Stationary points and critical points}\label{sec:statcrit}
\subsection{Stationary points}
We recall the definition of a stationary point of a function on a manifold, then discuss some basic facts about such points.

\begin{definition}\label{def:stat}
    Let $M$ be a real smooth manifold of dimension $k$ and $g:M\to\RR$ a smooth map. We say that $x\in M$ is a \textit{stationary point} of $g$ if the gradient of $g\circ\alpha^{-1}$ vanishes at $x$, where $\alpha:U\to\RR^k$ is a chart on an open subset $U\subset M$ containing $x$.
\end{definition}
Note that $x\in M$ is a stationary point of $g:M\to\RR$ if and only if $\frac{d}{dt}\big|_{t=0}g\circ \gamma(t)=0$ for any smooth curve $\gamma:\RR\to M$ such that $\gamma(0)=x$. Supposing that $\alpha$ is centered at $x$, this follows from the equation
\begin{equation}\label{eq:curve}
\frac{d}{dt}\Big|_{t=0} g\circ\gamma(t)=\frac{d}{dt}\Big|_{t=0} g\circ\alpha^{-1}\circ\alpha\circ\gamma(t)=\langle\nabla(g\circ\alpha^{-1})(0),(\alpha\circ\gamma)'(0)\rangle.
\end{equation}
If the gradient of $g\circ\alpha^{-1}$ vanishes, so does the whole expression. Conversely, if \eqref{eq:curve} vanishes for every choice of $\gamma$, then $\nabla(g\circ\alpha^{-1})=0$ by the nondegeneracy of the inner product.

Suppose $f:M\to\CC$ is complex-valued. \Cref{prop:power} tells us that the nonvanishing stationary points of $|f|$ are the same as those of $|f|^p$ for any nonzero $p\in\RR$. It will be useful for us to take $p=2$, since $|f|^2=f_1^2+f_2^2$ where $f_1$ and $f_2$ are real-valued functions such that $f=f_1+\im f_2$.

\begin{proposition}\label{prop:power}
    Let $M$ and $g:M\to\RR$ be as in \Cref{def:stat}. Suppose $x\in M$ such that $g(x)\neq 0$. Then $x$ is a critical point of $g$ if and only if $x$ is a critical point of $g^p$, where $0\neq p\in\RR$.
\end{proposition}
\begin{proof}
    Let $\alpha:U\to\RR^k$ be a chart centered at $x$ and let $F=g\circ\alpha^{-1}$. Then $\nabla(F^p)=pF^{p-1}\nabla F$. If $g(x)\neq 0$, then $\nabla(F^p)(0)=0$ if and only if $\nabla F(0)=0$.
\end{proof}

The set of stationary points is partitioned into equivalence classes by a natural group action, given in \Cref{prop:sgroup}. Since we can act on representatives, it suffices to find stationary points up to equivalence.

\begin{proposition}\label{prop:sgroup}
    Let $M$ and $g:M\to\RR$ be as in \Cref{def:stat}. Suppose $\Phi:M\to M$ is a smooth bijection such that $g=g\circ\Phi$. Then $x\in M$ is a stationary point of $g$ if and only if $\Phi(x)$ is a stationary point of $g$.
\end{proposition}
\begin{proof}
    Let $x\in M$. If $\alpha:U\to \RR^{k}$ is a chart centered at $x\in M$ then $\alpha\circ\Phi^{-1}:\Phi(U)\to\RR^k$ is a chart centered at $\Phi(x)$. Since $g=g\circ\Phi$, we have
    $\nabla(g\circ\alpha^{-1})(0)=\nabla(g\circ\Phi\circ\alpha^{-1})(0)$. Hence $x=\alpha^{-1}(0)$ is stationary if and only if $\Phi(x)=\Phi\circ\alpha^{-1}(0)$ stationary.
\end{proof}

\subsection{Complex and real Lie groups}\label{sec:real} We explain how to interpret $\SL_N(\CC)$ as a real Lie group; the facts presented here are standard.

Recall the identification $\CC^N=\RR^{2N}$ via the map $\varphi\mapsto\varphi'$ given in \eqref{eq:embed}. While it is true that $\langle \varphi,\varphi\rangle=\langle \varphi',\varphi'\rangle$ (so the induced norm is the usual norm), it is not generally true that $\langle\varphi,\psi\rangle=\langle\varphi',\psi'\rangle$ for $\varphi,\psi\in\CC^N$. The correct formula reads
\begin{equation}\label{eq:bracket}
\langle \varphi,\psi\rangle=\langle \varphi',\psi'\rangle + \im \langle J^{\oplus N}\varphi',\psi'\rangle, \quad J=\begin{pmatrix}
    0 & -1 \\ 1 & 0
\end{pmatrix}.
\end{equation}
Similarly, matrices in $\CC^{N\times N}$ embed into $\RR^{2N\times 2N}$ by the map
\[
M=
\begin{pmatrix}
    z_{11} & \dots & z_{1N} \\
    \vdots & & \vdots \\
    z_{N1} & \dots & z_{NN}
\end{pmatrix}
\mapsto M'=
\begin{pmatrix}
    U(z_{11}) & \dots & U(z_{1N}) \\
    \vdots & & \vdots \\
    U(z_{N1}) & \dots & U(z_{NN})
\end{pmatrix},
\]
where $U(z_{ij})=\begin{pmatrix} x_{ij} & -y_{ij} \\ y_{ij} & x_{ij} \end{pmatrix}$ and $x_{ij},y_{ij}$ are the real and imaginary parts of $z_{ij}=x_{ij}+\im y_{ij}$. This embedding respects our identification of $\CC^N$ with $\RR^{2N}$, that is, $(M\varphi)'=M'\varphi'$. Furthermore, the embedding respects the matrix exponential, i.e. it makes the following diagram commute.
\[\begin{tikzcd}
	{\mathfrak{sl}_N^{\mathbb{C}}} & {\mathfrak{sl}_{2N}^{\mathbb{R}}} \\
	{\operatorname{SL}_N(\mathbb{C})} & {\operatorname{SL}_{2N}(\mathbb{R})}
	\arrow[from=1-1, to=1-2]
	\arrow["\exp"', from=1-1, to=2-1]
	\arrow["\exp", from=1-2, to=2-2]
	\arrow[from=2-1, to=2-2]
\end{tikzcd}\]

\subsection{The critical set}\label{sec:critical}
As in the previous section, make the identification $(\CC^d)^{\otimes n}=\CC^N=\RR^{2N}$, where $N=d^n$, via the map $\varphi\mapsto\varphi'$ given in \eqref{eq:embed}. The embedding $\SL_N(\CC)\mapsto\SL_{2N}(\RR)$ restricts to an embedding $\rho(\SL_d(\CC)^{\times n})\to \SL_{2N}(\RR)$. Let $G_\CC=\SL_d(\CC)^{\otimes n}$ and let $G_\RR\subset \GL(\RR^N)$ be the corresponding real Lie group. We say that $\varphi\in(\CC^d)^{\otimes n}$ is \textit{critical} if $\langle X\varphi,\varphi\rangle=0$ for all $X\in\text{Lie}(G_\CC)$. It is equivalent to say that $\langle X'\varphi',\varphi'\rangle=0$ for all $X'\in\text{Lie}(G_\RR)$. To see this, specialize equation \eqref{eq:bracket}:
\[
\langle X\varphi,\varphi\rangle=\langle X'\varphi',\varphi'\rangle + \im\langle J^{\oplus N}X'\varphi',\varphi'\rangle.
\]
Observing that $J^{\oplus N}X'=(\im X)'\in \text{Lie}(G_\RR)$ if and only if $X\in\text{Lie}(G_\CC)$, the claim follows.

\begin{proposition}[Kempf-Ness]\label{prop:K-N}
The point $\varphi\in\mathcal{H}_n$ is critical if and only if $\|g.\varphi\|\geq \|\varphi\|$ for all $g\in \SL_2^{\times n}$. If $\varphi$ is critical and $g\in \SL_2^{\times n}$ is such that $\|g.\varphi\|= \|\varphi\|$, then there exists $h\in \SU_2^{\times n}$ such that $g.\varphi=h.\varphi$.
\end{proposition}
\begin{proof}
    See \cite[Theorem~2]{Gour_2011}.
\end{proof}

In other words, \Cref{prop:K-N} states that $\varphi$ is critical if and only if it minimizes the norm on its $\SL_d^{\times n}$-orbit. If $\psi$ is another minimizer in the $\SL_d^{\times n}$-orbit, then $\psi$ is in fact in the $\SU_d^{\times n}$-orbit of $\varphi$.
\begin{proposition}\label{prop:entropy}
    The point $\varphi\in\mathcal{H}_n$ is critical if and only if all single-qubit reduced density matrices of $\varphi$ are proportional to the identity.
\end{proposition}
\begin{proof}
    See \cite[Theorem~3]{Gour_2011}.
\end{proof}

Note that the von Neumann entropy of a reduced density matrix \cite{Nielsen_Chuang} is maximized if and only if the matrix is proportional to the identity. Together with \Cref{prop:entropy}, this justifies a claim made at the beginning of \Cref{sec:selection}.

Our primary interest in critical points is due to \Cref{prop:crit}. A consequence is that, to find the nonvanishing stationary points of $|f|$, it suffices to search over the set of critical points.
\begin{proposition}\label{prop:crit}
    Let $f:(\CC^d)^{\otimes n}\to (\CC^d)^{\otimes n}$ be a homogeneous $\SL_d^{\times n}$-invariant of degree $m>0$ and let $N=d^n$. Suppose $\varphi\in\SS^{2N-1}$ is a nonvanishing stationary point of $|f(x)|$ restricted to $\SS^{2N-1}$. Then $\varphi$ is critical.
\end{proposition}
\begin{proof}
    Consider $\varphi$ to be an element of the real vector space $\RR^{2N}$. We prove the contrapositive. Suppose $\varphi$ is not critical. Then there exists $X\in\text{Lie}(G_\RR)$ such that $\langle X\varphi,\varphi\rangle\neq 0$. Set $\gamma(t)=\exp(tX)\varphi$. We have
    \begin{align*}
        \frac{d}{dt}\Big|_{t=0}  \|\gamma(t) \|^2 &= \frac{d}{dt}\Big|_{t=0} \langle \exp(tX)\varphi,\exp(tX)\varphi\rangle \\
        &= \langle \exp(tX)X\varphi,\exp(tX)\varphi\rangle + \langle\exp(tX)\varphi,\exp(tX)X\varphi\rangle \Big|_{t=0}\\
        &= 2\langle \exp(tX)X\varphi,\exp(tX)\varphi\rangle \Big|_{t=0}\\
        &= 2\langle X\varphi,\varphi\rangle.
    \end{align*}
    If $f(\varphi)\neq 0$, then 
    \begin{align*}
        \frac{d}{dt}\Big|_{t=0} |f(\gamma(t)/\|\gamma(t)\|)| &= \frac{d}{dt}\Big|_{t=0}\|\gamma(t)\|^{-m}|f(\varphi)| \\
        &=|f(\varphi)|\frac{d}{dt}\Big|_{t=0}(\|\gamma(t)\|^2)^{-m/2} \\
        &=|f(\varphi)|( -m\|\varphi\|^{-m-2}\langle X\varphi,\varphi\rangle )\neq 0.
    \end{align*}
    Therefore $\varphi$ is not stationary.
\end{proof}

\section{Vinberg theory}\label{sec:vinberg}
\subsection{A graded Lie algebra}\label{sec:graded} The results we need from Vinberg's theory rely on the fact that the SLOCC module $\mathcal{H}:=\mathcal{H}_4$ embeds into the $\ZZ_2$-graded Lie algebra
$
\mathfrak{so}_8\cong\mathfrak{so}_4^{\times 2}\oplus\mathcal{H}.
$
We refer the reader to \cite{WallachGIT} for details. We now describe a construction of this embedding following Chterental and Djokovic \cite{ChtDjo:NormalFormsTensRanksPureStatesPureQubits}.

Each tensor $\varphi\in\mathcal{H}$ corresponds to the matrix $\tilde{\varphi}\in\CC^{4\times 4}$ of a linear map $(\CC^2\otimes\CC^2)^*\to\CC^2\otimes\CC^2$ associated to $\varphi$. Explicitly, the correspondence is
\begin{equation}\label{eq:flatten}
\varphi=\sum_{i,j,k,l=0}^1 a_{ijkl}\ket{ijkl}\quad\mapsto\quad \tilde{\varphi}=\begin{pmatrix}
    a_{0000} & a_{0001} & a_{0010} & a_{0011} \\
    a_{0100} & a_{0101} & a_{0110} & a_{0111} \\
    a_{1000} & a_{1001} & a_{1010} & a_{1011} \\
    a_{1100} & a_{1101} & a_{1110} & a_{1111}
\end{pmatrix}.
\end{equation}
Let $\kappa:\SL_4\times\SL_4\to\GL(\CC^{4\times 4})$ denote the representation defined by $\kappa(A,B)(\tilde{\varphi})=A\tilde{\varphi}B$. Then the induced $\SL_2^{\times 4}$-action on $\CC^{4\times 4}$ factors through $\SL_4\times\SL_4$, having the form
\[\begin{tikzcd}
	{\operatorname{SL}_2^{\times 4}} & {\operatorname{SL}_4\times\operatorname{SL}_4} & {\operatorname{GL}(\mathbb{C}^{4\times 4})}
	\arrow[from=1-1, to=1-2]
	\arrow["\kappa", from=1-2, to=1-3]
\end{tikzcd}\]
where the first map is $(A_1,A_2,A_3,A_4)\mapsto (A_1\otimes A_2,(A_3\otimes A_4)^\top)$. There exists a double cover $\SL_2\times\SL_2\to\SO_4$ defined using conjugation $(A_1, A_2)\mapsto T(A_1\otimes A_2)T^*$ by the unitary matrix
\[
T=\frac{1}{\sqrt{2}}\begin{pmatrix}
        1 & 0 & 0 & 1\\
        0 & \im & \im & 0\\
        0 & -1 & 1 & 0\\
        \im & 0 & 0 & -\im
    \end{pmatrix}.
\]
Applying the change of coordinates $\tilde{\varphi}\mapsto R_{\varphi}:= T\tilde{\varphi}T^*$, the induced $\SL_2^{\times 4}$-action on $\CC^{4\times 4}$ now factors through $\SO_4\times\SO_4$, having the form
\[\begin{tikzcd}
	{\operatorname{SL}_2^{\times 4}} & {\operatorname{SO}_4\times\operatorname{SO}_4} & {\operatorname{GL}(\mathbb{C}^{4\times 4})}
	\arrow[from=1-1, to=1-2]
	\arrow["\kappa", from=1-2, to=1-3]
\end{tikzcd}\]
where the first map is $(A_1,A_2,A_3,A_4)\mapsto (T(A_1\otimes A_2)T^*, T(A_3\otimes A_4)^\top T^*)$. Conjugation by the diagonal matrix $I_4\oplus(-I_4)$ decomposes the Lie algebra $\mathfrak{so}_8$ of skew-symmetric matrices into a direct sum of eigenspaces consisting of matrices with the following block formats:
\[
\mathfrak{so}_8=\left\{\begin{pmatrix}
    A & 0 \\
    0 & B
\end{pmatrix}: A,B\in\mathfrak{so}_4\right\}\oplus
\left\{\begin{pmatrix}
    0 & R_{\varphi} \\
    -R_{\varphi}^\top & 0
\end{pmatrix}:\varphi\in\mathcal{H}
\right\}.
\]
This establishes the decomposition $\mathfrak{so}_8 \cong\mathfrak{so}_4^{\times 2}\oplus\mathcal{H}$. A calculation shows
\[
\left[
\begin{pmatrix}
    A & 0 \\
    0 & B
\end{pmatrix},
\begin{pmatrix}
    0 & R_{\varphi} \\
    -R_{\varphi}^\top & 0
\end{pmatrix}
\right]
=
\begin{pmatrix}
    0 & A R_{\varphi} + R_{\varphi} B^\top \\
    -(A R_{\varphi} + R_{\varphi} B^\top)^\top & 0
\end{pmatrix}
\]
so the Lie algebra representation $\mathfrak{so}_4^{\times 2}\to\text{Aut}(\mathcal{H})$ obtained from the adjoint action corresponds to the natural Lie group representation $\SO_4^{\times 2}\to\GL(\mathcal{H})$.

\subsection{The Cartan subspace}
The semisimple elements in $\CC^{4\times 4}\subset\mathfrak{so}_8$ are those that are diagonalizable by the action of $\SO_4\times\SO_4$. We choose a so-called \textit{Cartan subspace} to be the 4-dimensional subspace of diagonal matrices in the new coordinates where the $\SL_2^{\times 4}$-action factors through $\SO_4^{\times 2}$. Translating back to the original coordinates, the Cartan $\mathfrak{a}\subset \mathcal{H}$ is the complex span of
\begin{align*}
    u_1 &= \frac{1}{2}(\ket{0000} + \ket{0011} + \ket{1100} +\ket{1111}), &
    u_2 &= \frac{1}{2}(\ket{0000} - \ket{0011} - \ket{1100} +\ket{1111}), \\
    u_3 &= \frac{1}{2}(\ket{0101} + \ket{0110} + \ket{1001} + \ket{1010}), &
    u_4 &= \frac{1}{2}(\ket{0101} - \ket{0110} - \ket{1001} + \ket{1010}).
\end{align*}
We will identify $\mathfrak{a}$ with the set of vectors $z\in\CC^4$, writing \[(z_1,z_2,z_3,z_4):=z_1u_1+z_2u_2+z_3u_3+z_4u_4.\]
A calculation shows that
\[
\|z_1u_1+z_2u_2+z_3u_3+z_4u_4\|^2=|z_1|^2+|z_2|^2+|z_3|^2+|z_4|^2,
\]
so the induced norm on $\CC^4$ is the usual Hermitian norm.

The following result, together with \Cref{prop:sgroup,prop:crit}, allows us to further restrict the search for stationary points to the linear subspace $\mathfrak{a}$.
\begin{proposition}\label{prop:critorbit}
    The set of critical points in $\mathcal{H}$ is the $\SU_2^{\times 4}$-orbit of $\mathfrak{a}$.
\end{proposition}
\begin{proof}
    This is a special case of \cite[Proposition~3.75]{WallachGIT}.
\end{proof}
\subsection{The Weyl group} The normalizer and centralizer of $\mathfrak{a}$ are defined as the subgroups
\begin{align*}
    N(\mathfrak{a})&:=\{g\in\SL_2^{\times 4}: g.\mathfrak{a}=\mathfrak{a}\}\quad\text{and} \\
    C(\mathfrak{a})&:=\{g\in\SL_2^{\times 4}: g.x=x \text{ for all } x\in\mathfrak{a}\}
\end{align*}
respectively. The \textit{Weyl group} of $\mathfrak{a}$ is the quotient $W:=N(\mathfrak{a})/C(\mathfrak{a})$. The group $W$ acts on the Cartan subspace $\mathfrak{a}$ like the set of linear maps
\begin{equation}\label{eq:weylg}
(z_1,z_2,z_3,z_4)\mapsto (\epsilon_1 z_{\pi(1)},\epsilon_2 z_{\pi(2)},\epsilon_3 z_{\pi(3)},\epsilon_4 z_{\pi(4)}),
\end{equation}
such that $\pi\in\mathfrak{S}_4$, $\epsilon_i=\pm 1$, and $\epsilon_1\epsilon_2\epsilon_3\epsilon_4=1$ (see \cite[Exercise~3, p.~179]{WallachGIT}).
\begin{proposition}\label{prop:restriction}
    The restriction map $f\mapsto f|_\mathfrak{a}$ for $f\in\CC[\mathcal{H}]$ induces an isomorphism of algebras $\CC[\mathcal{H}]^{\SL_2^{\times 4}}\cong \CC[z_1,z_2,z_3,z_4]^W.$
\end{proposition}
\begin{proof}
    This is a special case of \cite[Theorem~3.62]{WallachGIT}.
\end{proof}

\subsection{Critical states on the Cartan}\label{sec:nf} Given a critical point $\varphi\in\mathcal{H}$, we explain how to find a point $\varphi'\in\mathfrak{a}$ in the Cartan subspace equivalent to $\varphi$ under the joint action of the local unitary group $\U_2^{\times 4}$ and the symmetric group $\mathfrak{S}_4$.

As discussed in \Cref{sec:graded}, there is an equivariant isomorphism
\[
\mathcal{H}\to\CC^{4\times 4},\quad \varphi\mapsto R_\varphi
\]
where the representation $\SL_2^{\times 4}\to\GL(\CC^{4\times 4})$ factors through $\SO_4^{\times 2}$. The group $\SO_4^{\times 2}$ acts by left and right multiplication on the space of matrices $\CC^{4\times 4}$. Let $\{e_1,e_2,e_3,e_4\}$ be the standard basis for $\CC^4$ and $\{E_{ij}:1\leq i,j\leq 4\}$ be the basis of operators
$
E_{ij}e_k = \delta_{jk}e_i
$
for $\CC^{4\times 4}$. The basis vectors $u_1,\dots,u_4$ for $\mathfrak{a}$ correspond to the basis operators $E_{11},\dots,E_{44}$ which span the subspace of diagonal matrices. Thus, by \Cref{prop:critorbit}, every critical $\varphi\in\mathcal{H}$ corresponds to a matrix $R_\varphi$ that is $\SO_4^{\times 2}$-equivalent to a diagonal matrix. Define $\tau:\CC^{4\times 4}\to \CC^{4\times 4}$ by $\tau(R_\varphi)=R_\varphi R_\varphi^\top.$ The set of eigenvalues of $\tau(R_\varphi)$ is invariant under $\SO_2^{\times 4}$ since
\[
\tau(AR_\varphi B^\top) = AR_\varphi R_\varphi^\top A^\top = A \tau(R_\varphi) A^\top
\]
for all $(A,B)\in\SO_4^{\times 2}$. Observing that
\[\tau(\lambda_1 E_{11} + \lambda_2 E_{22} + \lambda_3 E_{33} + \lambda_4 E_{44})=\lambda_1^2 E_{11} + \lambda_2^2 E_{22} + \lambda_3^2 E_{33} + \lambda_4^2 E_{44},
\]
we obtain an algorithm for computing $\varphi'$ in two steps.
\begin{enumerate}
    \item[1.] Compute the eigenvalues $\mu_1,\mu_2,\mu_3,\mu_4$ of $\tau(R_\varphi)$.
    \item[2.] Set $\varphi' = \sqrt{\mu_1}u_1 + \sqrt{\mu_2} u_2 + \sqrt{\mu_3} u_3 + \sqrt{\mu_4}u_4$.
\end{enumerate}
The $\mathfrak{S}_4$-representation $\mathfrak{S}_4\to\GL(\mathcal{H})$ restricts to a subrepresentation $\pi:\mathfrak{S}_4\to\GL(\mathfrak{a})$ on $\mathfrak{a} \cong \CC^4$. Let $\sigma_i\in\mathfrak{S}_4$ denote the transposition $(i,i+1)$. The symmetric group $\mathfrak{S}_4$ is generated by $\sigma_1$, $\sigma_2$, and $\sigma_3$. Written as matrices, these have the form (also noted in \cite{GourWallach2014})
\begin{equation}\label{eq:sym4}
\pi(\sigma_1)=\pi(\sigma_3)=\begin{pmatrix}
    1 & 0 & 0 & 0 \\
    0 & 1 & 0 & 0 \\
    0 & 0 & 1 & 0 \\
    0 & 0 & 0 & -1
\end{pmatrix}\quad \text{and}\quad
\pi(\sigma_2)=\frac{1}{2}\begin{pmatrix}
    1 & 1 & 1 & 1 \\
    1 & 1 & -1 & -1 \\
    1 & -1 & -1 & 1\\
    1 & -1 & 1 & -1
\end{pmatrix}.
\end{equation}
Using the actions of $\pi(\sigma_1)$ \eqref{eq:sym4} and the  Weyl group \eqref{eq:weylg} we see that neither the order of the eigenvalues $\mu_1,\dots,\mu_4$ nor the choice of square roots are important in choosing $\varphi'$.

Using the algorithm described above and further simplifying by the actions of the Weyl group, the symmetric group, and the circle group
$
\mathbb{T}:=\{e^{\im t}:t\in\RR\}
$
(which acts by scalar multiplication), we find that the states introduced in \Cref{sec:selection} have the following representatives in $\mathfrak{a}$. For ease of notation, we do not normalize the vectors.
\begin{align*}
    \ket{GHZ}&\cong(1,1,0,0), & \ket{MP}&\cong(1,0,0,0), \\
    \ket{YC}\cong\ket{C_1}\cong\ket{OS}&\cong(1,\im,0,0), &\ket{HS}\cong\ket{M}&\cong(\sqrt{3},\im,\im,\im), \\
    \ket{HD}\cong\ket{L}&\cong(1,e^{\pi\im/3},e^{2\pi\im/3},0), & \ket{BSSB}&\cong(1,\im,e^{-\pi \im/4},e^{\pi \im/4}).
\end{align*}
It was previously observed \cite{Enriquez_2016} that $\ket{HS}$ and $\ket{M}$ are equivalent. The fact that $\ket{HD}$ and $\ket{L}$ are equivalent is no surprise either, since it was claimed (likely based on numerical computations) in \cite{OsterlohHD} that $\ket{HD}$ is the unique global maximizer of the norm of the four-qubit hyperdeterminant, while the same thing was conjectured about $\ket{L}$ in \cite{GourWallach2014}. This conjecture was later proved by Chen and \DJ okovi\'c \cite{PhysRevA.88.042307}.


\subsection{SLOCC invariants}
The algebra $\CC[z_1,z_2,z_3,z_4]^W$ is freely generated by the invariants
\[
\mathcal{E}_0(z) = z_1 z_2 z_3 z_4\quad\text{and}\quad \mathcal{E}_i(z) = z_1^{2i} + z_2^{2i} + z_3^{2i} + z_4^{2i}\text{ for } 1\leq i\leq 3.
\]
Identifying $\mathcal{H}$ with the space of matrices $\tilde{\varphi}\in\CC^{4\times 4}$ via the map \eqref{eq:flatten}, we define the following polynomials (which appear in \cite{GourWallach2010}):
\[
\mathcal{G}_0(\tilde{\varphi}) = \det(\tilde{\varphi})\quad\text{and}\quad \mathcal{G}_i(\tilde{\varphi})=\text{tr}((\tilde{\varphi}(J\otimes J)\tilde{\varphi}^\top (J\otimes J))^i)\text{ for } 1\leq i\leq 3
\]
with $J$ as in \eqref{eq:bracket}. The polynomials $\mathcal{G}_i$ are $\SL_2^{\times 4}$-invariant due to the fact that $\SL_2$ preserves the symplectic bilinear form associated to $J$. These $\SL_2^{\times 4}$-invariants correspond to the $W$-invariants via the restriction map, i.e. $\mathcal{G}_i|_\mathfrak{a}=\mathcal{E}_i$ for all $0\leq i\leq 4$. Note that $\mathcal{G}_1$ is the unique $\SL_2^{\times 4}$-invariant in degree 2 and $|\mathcal{G}_1|^2$ coincides with the Christensen-Wong 4-tangle \cite{WongNelson01}.

\subsection{Symmetric invariants}\label{sec:sym} We are particularly interested in those SLOCC invariants which are also symmetric. The algebra consisting of such polynomials is known \cite{GourWallach2014} to be generated by $\mathcal{F}_1$, $\mathcal{F}_3$, $\mathcal{F}_4$, $\mathcal{F}_6$ whose restrictions have the form 
\[
\mathcal{F}_k|_{\mathfrak{a}}(z) = \sum_{i<j}(z_i-z_j)^{2k} + \sum_{i<j}(z_i+z_j)^{2k}, \text{ for } k\in\{1,3,4,6\}.
\]
Written in terms of the polynomials $\mathcal{E}_i$, these read
\begin{align*}
    \mathcal{F}_1|_{\mathfrak{a}}&=6\mathcal{E}_1,\\
\mathcal{F}_3|_{\mathfrak{a}} &= 30\mathcal{E}_1\mathcal{E}_2-24\mathcal{E}_3,\\
\mathcal{F}_4|_{\mathfrak{a}} &= -20\mathcal{E}_1^4 + 120\mathcal{E}_1^2 \mathcal{E}_2 + 480 \mathcal{E}_0^2 + 10\mathcal{E}_2^2 - 104\mathcal{E}_1\mathcal{E}_3,
\end{align*}
and
\begin{multline*}
\mathcal{F}_6|_{\mathfrak{a}} = -148\mathcal{E}_1^6  + 565\mathcal{E}_1^4 \mathcal{E}_2 + 5460\mathcal{E}_0^2 \mathcal{E}_1^2  + 540\mathcal{E}_1^2 \mathcal{E}_2^2 \\  - 570\mathcal{E}_1^3 \mathcal{E}_3 + 2160\mathcal{E}_0^2 \mathcal{E}_2 - 15\mathcal{E}_2^3  - 610\mathcal{E}_1\mathcal{E}_2\mathcal{E}_3 + 244\mathcal{E}_3^2.
\end{multline*}
Since the restriction map is an isomorphism (\Cref{prop:restriction}), one obtains corresponding expressions for the generators $\mathcal{F}_k$ by substituting $\mathcal{E}_i=\mathcal{G}_i$ for $0\leq i\leq 3$.

\subsubsection{Hyperdeterminant} Miyake \cite{Miyake_2002} gives an expression for the four-qubit hyperdeterminant restricted to $\mathfrak{a}$. In our basis, it reads
\[
\text{Hdet}|_\mathfrak{a}(z)= \prod_{i< j} (z_i^2-z_j^2)^2.
\]
In terms of the symmetric generators $\mathcal{G}_1=\frac{1}{6}\mathcal{F}_1$, $\mathcal{F}_3$, $\mathcal{F}_4$, and $\mathcal{F}_6$, the hyperdeterminant is
\begin{multline*}
    \text{Hdet} = 314928000^{-1} (-23794560\mathcal{G}_1^{12}   + 14450400\mathcal{G}_1^9 \mathcal{F}_3 - 6828300\mathcal{G}_1^8 \mathcal{F}_4 - 2211120\mathcal{G}_1^6 \mathcal{F}_3^2 \\ + 2043360\mathcal{G}_1^5 \mathcal{F}_3 \mathcal{F}_4 + 563760\mathcal{G}_1^6 \mathcal{F}_6 + 5376\mathcal{G}_1^3 \mathcal{F}_3^3  - 484380\mathcal{G}_1^4 \mathcal{F}_4^2  + 6552\mathcal{G}_1^2 \mathcal{F}_3^2 \mathcal{F}_4 \\ - 172800\mathcal{G}_1^3 \mathcal{F}_3 \mathcal{F}_6 - 40\mathcal{F}_3^4  - 5832\mathcal{G}_1 \mathcal{F}_3 \mathcal{F}_4^2  + 81000\mathcal{G}_1^2 \mathcal{F}_4 \mathcal{F}_6 + 729\mathcal{F}_4^3  + 720\mathcal{F}_3^2 \mathcal{F}_6 - 3240\mathcal{F}_6^2).
\end{multline*}
Holweck and Oeding give another expression for this hyperdeterminant using a different set of generators \cite{HolweckOedingE8}.

\section{Stationary points of symmetric invariants}\label{sec:compute}
\subsection{Systems of equations} Let $g:\RR^k\to\RR$ be a homogeneous polynomial. Consider the problem of finding nonvanishing stationary points $\varphi\in\SS^{k-1}$ of the restriction of $g$ to $\SS^{k-1}$. By the method of Lagrange multipliers, $\varphi$ is a stationary point of $g$ if and only if
\begin{equation}\label{eq:lagrange}
\nabla g (\varphi)= \lambda (x_1,x_2,\dots,x_{k})\Big|_{x=\varphi}
\end{equation}
for some $\lambda\in\RR$. Indeed, the condition above says that $\nabla g(\varphi)\in (T_\varphi \SS^{k-1})^\perp$. This is equivalent to $\langle \nabla g(\varphi), u\rangle=0$ for every $u\in T_\varphi \SS^{k-1}$, which is equivalent to $\frac{d}{dt}\big|_{t=0} g(\gamma(t))=0$ for every smooth $\gamma:\RR\to\SS^{k-1}$ such that $\gamma(0)=\varphi$.

We can turn \eqref{eq:lagrange} into a system of equations by taking all $2\times 2$ minors of the $2\times k$ matrix obtained from the vectors that appear in the equation, since the rows of the matrix are proportional if and only if the matrix is rank 1. It is enough to take $k-1$ of these matrix minors, reasoning as follows. For $j=1,\dots,k$ let $U_j\subset\SS^{k-1}$ be the dense open set consisting of $\varphi\in\SS^{k-1}$ such that $x_j(\varphi)\neq 0$. Then $\varphi\in U_j$ is stationary if and only if it solves
\begin{equation}\label{eq:sys}
    \frac{\partial g}{\partial x_i}x_{j} - \frac{\partial g}{\partial x_{j}}x_i=0,\quad i\neq j.
\end{equation}
Since $g$ is a homogeneous polynomial, the the system \eqref{eq:sys} consists of homogeneous polynomials and its solution set in the ambient space $\RR^k$ is a union of lines. Any point on a line can easily be normalized to get a corresponding point on $\SS^{k-1}$. Therefore, to find nonvanishing stationary points $\varphi\in U_j$, we may  set $x_j=1$ and solve the dehomogenized system
\begin{equation}\label{eq:dehom}
\left(\frac{\partial g}{\partial x_i}- \frac{\partial g}{\partial x_{j}}x_i\right)\Big|_{x_j=1}=0,\quad i\neq j.
\end{equation}
We are particularly interested in the case where $g:\CC^r\to\RR$ has the form $g(z)=|f(z)|^2$, where $f\in\CC[z_1,\dots,z_r]$ is a complex homogeneous polynomial of degree $m>0$. Note that $f(e^{\im t}z)=e^{\im mt}f(z)$, hence $g$ is invariant under multiplication by elements of the circle group $\mathbb{T}$. By \Cref{prop:sgroup}, the solutions of \eqref{eq:sys} in $\SS^{2r-1}\subset\CC^r$ is a union of $\mathbb{T}$-orbits, hence the variety has no 0-dimensional components. To fix this issue, we add the constraint $f(\varphi)\in\RR$. In this case, it suffices to solve the system \eqref{eq:sys} or \eqref{eq:dehom} substituting $g$ with the real part of the complex polynomial $f$. The following proposition explains.

\begin{proposition}\label{prop:realpart}
    Let $f=f_1+\im f_2$ be a polynomial in $\CC[z_1,\dots,z_r]$ with real and imaginary parts $f_1$ and $f_2$ respectively. Suppose $\varphi\in\SS^{2r-1}\subset\CC^{r}$ such that $f(\varphi)$ is real. Then $\varphi$ is a nonvanishing stationary point of $|f(z)|^2$ restricted to $\SS^{2r-1}$ if and only if $\varphi$ is a nonvanishing stationary point of $f_1$ restricted to $\SS^{2r-1}$.
\end{proposition}
\begin{proof}
    Choose a chart $\alpha:U\to\RR^{2r-1}$ on an open set $U\subset\SS^{2r-1}$ centered at $\varphi$. Let $F:\RR^{2r-1}\to \RR$ be the composition \[F(x)=|f(\alpha^{-1}(x))|^2=f_1(\alpha^{-1}(x))^2+f_2(\alpha^{-1}(x))^2.\]
    Since $f(\varphi)$ is real, $f_2(\varphi)=0$. Hence
    \begin{align*}
    \frac{\partial F}{\partial x_i}(0) &=2f_1(\varphi)\frac{\partial}{\partial x_i}(f_1\circ\alpha^{-1})(0)+2f_2(\varphi)\frac{\partial}{\partial x_i}(f_2\circ\alpha^{-1})(0) \\
    &= 2f_1(\varphi)\frac{\partial}{\partial x_i}(f_1\circ \alpha^{-1})(0).
    \end{align*}
    Then the gradient at $\varphi$ is
    $
        \nabla F(0)
        = 2f_1(\varphi)\nabla (f_1\circ\alpha^{-1})(0).
    $
    If $f_1(\varphi)\neq 0$, then $\nabla F(0)=0$ if and only if $\nabla (f_1\circ\alpha^{-1})(0)=0$.
\end{proof}

\subsection{Four-qubit stationary points} We return to the four-qubit SLOCC module $\mathcal{H}=\mathcal{H}_4$ and continue using the notation and definitions introduced in \Cref{sec:vinberg}. We also define $\SS^{15}\subset\mathcal{H}$ to be the set of unit vectors in $\mathcal{H}$, $\SS^{7}:=\SS^{15}\cap \mathfrak{a}$ to be the set of unit vectors in the Cartan subspace, and $\mathbb{T}$ to be the circle group. Let $f\in\CC[\mathcal{H}]$ be a homogeneous symmetric SLOCC invariant polynomial of degree $m>0$, and define $F:\mathcal{H}\to\RR$ by $F(z)=|f(z)|^2$. Our goal is to find the nonvanishing stationary points of $F|_{\SS^{15}}$.
\subsubsection{Reducing to the Cartan subspace} Let $S\subset\CC[\mathcal{H}]$ denote the set of homogeneous symmetric SLOCC invariant polynomials. Define
\begin{equation}\label{eq:G}
\tilde{G}:=\{\text{smooth }\Phi:\SS^{15}\to\SS^{15}\:|\: \text{$|f\circ\Phi|=|f|$ for all $f\in S$}\}.
\end{equation}
Recall, by \Cref{prop:sgroup}, that $\tilde{G}$ preserves the stationary points of $F|_{\SS^{15}}$. By \Cref{prop:crit}, every nonvanishing stationary point $\varphi\in\SS^{15}$ is critical. By \Cref{prop:critorbit}, the set of critical points in $\SS^{15}$ is the $\SU_2^{\times 4}$-orbit of $\SS^{7}$. Since $\rho(\SU_2^{\times 4})\subset \tilde{G}$, where $\rho:\GL_2^{\times 4}\to\GL(\mathcal{H})$ is the natural representation  defined in \eqref{eq:rep}, we can find every nonvanishing critical point of $F$ on $\SS^{15}$ up to $\tilde{G}$-equivalence by searching over points in $\SS^7$.

\subsubsection{Symmetries of stationary points}\label{sec:wsym} Now define
\[
\tilde{W}:=\{\text{smooth }\Phi:\SS^{7}\to\SS^{7}\:|\: \text{$|f\circ\Phi|=|f|$ for all $f\in S$}\}.
\]
By \Cref{prop:sgroup}, $\tilde{W}$ preserves the stationary points of $F|_{\SS^{7}}$. Since the Weyl group arises from restricting the $\SL_2^{\times 4}$-action, we have $\rho'(W)\subset\tilde{W}$, where $\rho'$ is the homomorphism induced by $\rho$. 
Since $f$ is symmetric, $\pi(\mathfrak{S}_4)\subset \tilde{W}$, where $\pi:\mathfrak{S}_4\to\GL(\mathfrak{a})$ is the representation given explicitly in \eqref{eq:sym4}. Finally, we note that entrywise complex conjugation $\varphi\mapsto\varphi^*$ is also an element of $\tilde{W}$.

\subsubsection{Stationary points of $|\mathcal{F}_1|$} The invariant $\mathcal{F}_1=6\mathcal{G}_1$ is simple enough that we can find all of its stationary points of $|\mathcal{F}_1|$ analytically. The following proposition summarizes the situation.

\begin{proposition}
    The state $\varphi\in \SS^{15}\subset\mathcal{H}$ is a global maximizer of the 4-tangle $|\mathcal{F}_1|^2$ restricted to $\SS^{15}$ if and only if $\varphi$ is in the local unitary orbit of the set $\SS^{7}_\RR\subset\SS^7$ of real unit vectors in the Cartan subspace. There are no other nonvanishing stationary points.
\end{proposition}
\begin{proof} 
Let $F:\mathcal{H}\to\RR$ be the map $z\mapsto |\mathcal{F}_1(z)|$ and define $\tilde{G}$ as in \eqref{eq:G}. Make the identification $\mathfrak{a}=\CC^4=\RR^8$ with the embedding \eqref{eq:embed}. The real part of $\mathcal{F}_1|_{\mathfrak{a}}$ is
\begin{equation}\label{eq:gf1}
\text{Re}(\mathcal{F}_1|_{\mathfrak{a}}(z))=6(x_1^2 - y_1^2 + x_2^2 - y_2^2 + x_3^2 - y_3^2 + x_4^2 - y_4^2).
\end{equation}
The ideal generated by the dehomogenized polynomials \eqref{eq:dehom} with $g$ set equal to the right side of \eqref{eq:gf1} is $\langle y_1,y_2,y_3,y_4\rangle$. By \Cref{prop:realpart}, if $\varphi\in\SS^7$ is a nonvanishing stationary point of $F|_{\SS^7}$ such that $\mathcal{F}_1(\varphi)$ is real and $x_1(\varphi)\neq 0$, then $\varphi$ is real. Repeating the calculation with different orderings on the variables, we conclude that if $\varphi\in\SS^7$ is a nonvanishing stationary point of $F|_{\SS^7}$ such that $\mathcal{F}_1(\varphi)$ is real, then $\varphi$ is purely real or purely imaginary. Given any $\psi\in\SS^7$ there exists $\varphi$ in the $\mathbb{T}$-orbit of $\psi$ such that $\mathcal{F}_1(\varphi)$ is real. Since $\psi\mapsto e^{\im t}\psi$ is an element of $\tilde{G}$, we showed that
\begin{equation}\label{eq:1}
X\cap\SS^7\subset\mathbb{T}(\SS^7_\RR\cup\im\SS^7_\RR)=\mathbb{T}\SS^7_\RR,
\end{equation}
where $X\subset\SS^{15}$ is the set of nonvanishing stationary points of $F|_{\SS^{15}}$. By \Cref{prop:crit,prop:critorbit},
\begin{equation}\label{eq:2}
X\subset(\SU_2^{\times 4}\mathfrak{a})\cap\SS^{15}=\SU_2^{\times 4}\SS^7.
\end{equation}
Since $\SU_2^{\times 4}$ acts on $\SS^{15}$ like elements of $\tilde{G}$, $X$ is a union of $\SU_2^{\times 4}$-orbits. By \eqref{eq:2}, each $\SU_2^{\times 4}$-orbit intersects $\SS^7$; by \eqref{eq:1}, the intersection consists of elements in $\mathbb{T}\SS^7_\RR$. It follows that
\[
X\subset\SU_2^{\times 4}(\mathbb{T}\SS^7_\RR)=\U_2^{\times 4}\SS^7_\RR.
\]
Since $\mathcal{F}_1(\varphi)=6\|\varphi\|^2=6$ whenever $\varphi\in\SS^7_\RR$, we have $F(\varphi)=6$ whenever $\varphi\in \U_2^{\times 4}\SS^7_\RR$. By compactness, $F|_{\SS^{15}}$ has a global maximizer, which is a nonvanishing stationary point. Therefore, the global maximum is 6 and $X=\U_2^{\times 4}\SS^7_\RR$ is the set of global maximizers.
\end{proof}

\begin{table}[t]
\renewcommand{\arraystretch}{1.2}
    \begin{tabular}{l|l}
    $\varphi_1=(1,0,0,0)\cong \ket{MP}$ &
    $\varphi_2=(1,1,0,0)\cong\ket{GHZ}$\\
    $\varphi_3=(1,1,1,0)$ &
    $\varphi_4=(2,1,1,0)$ \\
    $\varphi_5=(\sqrt{2},\im,0,0)$ &
    $\varphi_6=(\sqrt{2},\im,\im,0)$\\
    $\varphi_7=(\sqrt{2},\im,\im,\im)$& 
    $\varphi_8=(\sqrt{3},\im,\im,\im)\cong\ket{HS}$\\
    $\varphi_9=(1,e^{\pi\im/3},e^{2\pi\im/3},0)\cong\ket{HD}$&
    $\varphi_{10}=(\sqrt{2},\sqrt{2}-\sqrt{3}+\im,\sqrt{3}-\sqrt{2}+\im,\sqrt{6}-2\im)$\\
    $\varphi_{11}=(7,2\sqrt{7}\im,2\sqrt{7}\im,0)$&
    $\varphi_{12}=(18,11-\sqrt{203}\im,7+\sqrt{203}\im,0)$ 
    \\
\multicolumn{2}{l}{$\varphi_{13}=(1,a+b\im,-a-b\im,0)$, where 
    $a \approx 0.0933$, $b\approx 0.622 $
        satisfy} 
        \\
\multicolumn{2}{c}{
 \begin{tabular}{rl}
$1000a^{2}b^{2}-872b^{4}+85a^{2}+345b^{2}-7$&$=0$, \\
$100a^{4}-244b^{4}+45a^{2}+65b^{2}+11$&$=0$.
\end{tabular}}
    \\ 
\multicolumn{2}{l}{$\varphi_{14}=(1,a-b\im,a+b\im,c\im)$, where 
    $a \approx 0.217$, $b\approx 0.830$, $c \approx 0.366 $ 
    satisfy} \\
\multicolumn{2}{c}{
    \begin{tabular}{rl}
$5a^4 c-30a^2 b^2 c+5b^4 c-5a^2 c^3+5b^2 c^3+2c^5+5a^2 c-5b^2 c-5c^3+2c$ &$=0$, \\
$10a^4b-40a^2b^3+14b^5-15a^2bc^2+5b^3c^2+5bc^4+25a^2b-15b^3-5bc^2+4b$ &$=0$, \\
$4a^5+20a^3b^2-5a^3c^2+15ab^2c^2-5a^3-5ab^2+5ac^2+a $&$=0$.
\end{tabular}}
\end{tabular}
    \caption{Stationary points of $|\mathcal{F}_3|$.}
    \label{tab:statF3}
\end{table}

\subsubsection{Stationary points of $|\mathcal{F}_3|$}\label{sec:F3computation} To find stationary points of $|\mathcal{F}_3|$, we set up a system of dehomogenized equations as in \eqref{eq:dehom} with $g$ equal to the real part of $\mathcal{F}_3$. This system consists of seven polynomial equations (inhomogeneous) of degree 6.  We solved the system using the Homotopy Continuation library in Julia \cite{hc.jl}, which tracks 131,505 paths (the mixed volume) and finds (after about 15 minutes of computation) 9,471 nonsingular real points, which are highly accurate approximations of solutions. From these approximations we hypothesized a real affine subspace in which the exact solutions lie. Then we performed symbolic computations with Macaulay2 \cite{M2} to find expressions or polynomial relations for these solutions, verifying that they are stationary on $\SS^7$. This method has an advantage over just attempting to decompose the original ideal since the additional constraints simplify the polynomials and make symbolic methods tractable.

After reducing by the action of the group $\tilde{W}$ (see \Cref{sec:wsym}), we find 14 distinct equivalence classes listed in \Cref{tab:statF3} with representatives $\varphi_i$ for $1\leq i\leq 14$.  For ease of notation, we do not normalize the vectors. We choose representatives which have a real first coordinate; this can always be done by the action of $\mathbb{T}$. The $\varphi_i$ are in distinct $\tilde{G}$-orbits because they have distinct images under the map \[\varphi_i\mapsto (|\mathcal{F}_1(\varphi_i)|,|\mathcal{F}_3(\varphi_i)|,|\mathcal{F}_4(\varphi_i)|,|\mathcal{F}_6(\varphi_i)|).\] Using the polynomials in \Cref{sec:sym}, we check numerically that these points are stationary for $\|\mathcal{F}_3\|$ restricted to $\mathbb{S}^{15}$ and not just restricted to $\mathbb{S}^{7}$. Note that $\varphi_1$, $\varphi_2$, $\varphi_8$, and $\varphi_9$ are $\tilde{G}$-equivalent to states that appear in \Cref{sec:selection}.

\begin{table}[t]
\renewcommand{\arraystretch}{1.2}
    \begin{tabular}{l|l}
    $\psi_1=(1,0,0,0)\cong\ket{MP}$ &
    $\psi_2=(1,1,0,0)\cong\ket{GHZ}$\\
    $\psi_3=(1,1,1,0)$ &
    $\psi_4=(2,1,1,0)$ \\
    $\psi_5=(1,\im,0,0)\cong\ket{C_1}$ &
    $\psi_6=(1,\im,e^{-\pi \im/4},e^{\pi \im/4})\cong\ket{BSSB}$\\
    $\psi_7=(\sqrt{33},\sqrt{33},\sqrt{13}-\sqrt{20}\im,\sqrt{13}-\sqrt{20}\im)$\\
    \multicolumn{2}{l}{$\psi_8=(1,a \im,a \im,0)$, where $a\approx 0.393$
    satisfies} \\
    \multicolumn{2}{c}{
    $80a^6-91a^4+77a^2-10=0.$
    }
    \\
    \multicolumn{2}{l}{$\psi_9=(1,a \im,a \im,a \im)$, where $a\approx 0.342$
    satisfies} \\
    \multicolumn{2}{c}{
    $75a^6-63a^4+49a^2-5=0$.
    }
    \\
    \multicolumn{2}{l}{$\psi_{10}=(1,\frac{1}{2}+a \im,\frac{1}{2}-a \im,0)$, where 
    $a \approx 0.518$
    satisfies} \\
    \multicolumn{2}{c}{
    $640a^6-1232a^4+2016a^2-465=0$.
    }\\
    $\psi_{11}=(1,a \im,b,0)$, where 
    $a \approx 0.920$, $b \approx 0.302$
    satisfy  \\
\multicolumn{2}{c}{
 \begin{tabular}{rl}
$35a^4 - 21a^2 b^2 - 4b^4 - 21a^2 - 18b^2 - 4$ &$= 0$, \\
$190a^2b^4 - 110b^6 - 237a^2b^2 - 339b^4 + 190a^2 - 339b^2 - 110 $&$=0 $.  
    \end{tabular}}\\
    \multicolumn{2}{l}{ $\psi_{12}=(1,a \im,b,b)$ where 
    $a\approx 0.879$, $ b\approx 0.256$
    satisfy} \\
\multicolumn{2}{c}{
 \begin{tabular}{rl}
        $35a^4 - 21a^2b^2 + 52b^4 - 21a^2 + 3b^2 - 4 $&$= 0$, \\
$587a^2b^4 + 446b^6 - 918a^2b^2 + 129b^4 + 475a^2 - 732b^2 - 275 $&$=0$.
    \end{tabular}}\\
    \multicolumn{2}{l}{$\psi_{13}=(1,a+b \im,a+b \im,0)$, where $a\approx 0.347$, $b\approx 0.716$ 
    satisfy} \\
 \multicolumn{2}{c}{
 \begin{tabular}{rl}
$2696a^{6}-24472a^{2}b^{4}+7792b^{6}-161a^{4}-2030a^{2}b^{2} -4221b^{4}+679a^{2}+1659b^{2}+26$ &$=0$,\\
$94360a^{4}b^{2}-213584a^{2}b^{4}+51096b^{6}+3437a^{4}+2310a^{2}b^{2} -29687b^{4}+1505a^{2}$\\$+13405b^{2}-262$ &$=0$.
\end{tabular}}
\end{tabular}
    \caption{Stationary points of $|\mathcal{F}_4|$.}
    \label{tab:statF4}
\end{table}

\subsubsection{Stationary points of $|\mathcal{F}_4|$}\label{sec:F4computation} Repeating the process of \Cref{sec:F3computation}, we obtain a system of 7 polynomial equations (inhomogeneous) of degree 8 in 7 variables. We experimented with a variety of methods to solve this system, such as primary decomposition in symbolic algebra software \cite{M2}, Newton's method for many initial points, and numerical irreducible decomposition or solving with regeneration in Bertini \cite{BertiniSoftware}. Ultimately, the method that worked was to perform a standard solve command with HomotopyContinuation.jl \cite{hc.jl} on a 2020 Mac (3.3 GHz 6-Core Intel Core i5 with 72GB of RAM). This computation automatically used a polyhedral start system, with the mixed volume $1,367,387$, rather than the naive Bézout bound $8^7 =2,097,152$. After approximately 7 hours of compute time the result was a set of $936,047$ nonsingular solutions, $10,963$ of them real.

We are grateful to Jon Hauenstein who was able to run our code in Bertini in parallel on his cluster which found a result (in approximately one day) of 10,971 real solutions.

\begin{remark}The discrepancy in the number of solutions from different numerical computations applied to the same system is not uncommon and illustrates the fact that solutions can be lost with large computations. We cannot be sure if a numerical computation produces a complete set of solutions. However, this is not a huge concern for us since we only need to find one representative from each $\tilde{W}$-orbit. Moreover, if a solution is missing, then all solutions in its $\tilde{W}$-orbit must have been missed. For our computations of stationary points of $|\mathcal{F}_4|$ both sets of numerical solutions produce the same set of representatives.
\end{remark}
Like in the previous case, we used the numerical results to hypothesize the exact forms of each representative, then symbolically verified that the forms we report are stationary points on $\mathbb{S}^7$. We check numerically that these points are stationary on $\SS^{15}$. This produced 13 stationary points of $|\mathcal{F}_4|$ which belong to distinct $\tilde{G}$-orbits, listed in \Cref{tab:statF4} as $\psi_i$ for $1\leq i\leq 13$. As before, we do not normalize the vectors. Note that the first four states $\psi_1,\dots,\psi_4$ in \Cref{tab:statF4} are the same as the first four states $\varphi_1,\dots,\varphi_4$ in \Cref{tab:statF3}. Furthermore, $\psi_5$ and $\psi_6$ are $\tilde{G}$-equivalent to states that appear in \Cref{sec:selection}.

\begin{table}[b]
    \centering
    \begin{tabular}{c|c|c|c|c|c|c|c}
         & $|\mathcal{F}_1|$ & $|\mathcal{F}_3|$ & $|\mathcal{F}_4|$ & $|\mathcal{F}_6|$ & $\textbf{H}(|\mathcal{F}_3|)$ & $\textbf{H}(|\mathcal{F}_4|)$ & $\textbf{H}(|\mathcal{F}_6|)$ \\
        \hline
        $\varphi_1^*$ & 6 & 6 & 6 & 6 & (3,1,3) & (3,1,3) & (3,1,3) \\
        $\varphi_2^*$ & 6 & 9 & 16.5 & 64.125 & (6,1,0) & (6,1,0) & (6,1,0) \\
        $\varphi_3$ & 6 & $7.\overline{3}$ & $9.\overline{5}$ & $\sim$16.9 & (4,1,2) & (4,1,2) & (4,1,2) \\
        $\varphi_4$ & 6 & $7.\overline{6}$ & $10.7\overline{2}$ & $\sim$23.0 & (5,1,1) & (5,1,1) & (5,1,1) \\
        $\varphi_5$ & 2 & $0.\overline{6}$ & $1.\overline{259}$ & $1.\overline{259}$ & (3,1,3) & n/a & n/a  \\
        $\varphi_6$ & 0 & 2.25 & 1.40625 & $\sim$1.35 & (4,1,2) & n/a & n/a  \\
        $\varphi_7$ & 1.2 & 2.64 & 1.392 & 0.912768 & (5,1,1) & n/a & n/a  \\
        $\varphi_8^*$ & 0 & $2.\overline{6}$ & 0 & $\sim$0.790 & (6,1,0) & n/a  & n/a  \\
        $\varphi_9^*$ & 0 & $2.\overline{6}$ & 0 & $\sim$3.01 & (6,1,0) & n/a & n/a \\
        $\varphi_{10}$ & 0 & $\sim$2.45 & 0 & $1.6041\overline{6}$ & (5,1,1) & n/a & n/a \\
        $\varphi_{11}$ & 0.4 & $2.24\overline{8}$ & $\sim$1.93 & $\sim$1.55 & (3,1,3) & n/a & n/a \\
        $\varphi_{12}$ & 0.96 & 2.6496 & 1.8223104 & $\sim$3.23 & (5,1,1) & n/a & n/a\\
        $\varphi_{13}$ & $\sim$1.13 & $\sim$2.23 & $\sim$1.20 & $\sim$1.23 & (3,1,3) & n/a & n/a \\
        $\varphi_{14}$ & $\sim$0.963 & $\sim$2.43 & $\sim$1.41 & $\sim$1.74 & (4,1,2) & n/a & n/a \\
        $\psi_{5}^*$ & 0 & 0 & 2.5 & 1.875 & n/a & (5,1,1) & (5,1,1) \\
        $\psi_{6}^*$ & 0 & 0 & 1.875 & 0 & n/a & (3,1,3) & n/a \\
        $\psi_{7}$ & $\sim$3.77 & $\sim$0.342 & $0.\overline{06}$ & $\sim$2.36 & n/a & (3,1,3) & n/a \\
        $\psi_{8}$ & $\sim$3.17 & $\sim$0.932 & $\sim$1.66 & $\sim$0.00699 & n/a & (3,1,3) & n/a \\
        $\psi_{9}$ & $\sim$2.88 & $\sim$1.47 & $\sim$2.41 & $\sim$1.30 & n/a & (3,1,3) & n/a \\
        $\psi_{10}$ & $\sim$2.84 & $\sim$1.34 & $\sim$3.87 & $\sim$4.65 & n/a & (5,1,1) & n/a \\
        $\psi_{11}$ & $\sim$0.758 & $\sim$0.440 & $\sim$2.27 & $\sim$0.834 & n/a & (4,1,2) & n/a \\
        $\psi_{12}$ & $\sim$1.13 & $\sim$0.624 & $\sim$2.17 & $\sim$0.311 & n/a & (3,1,3) & n/a \\
        $\psi_{13}$ & $\sim$2.69 & $\sim$1.33 & $\sim$3.01 & $\sim$3.46 & n/a & (4,1,2) & n/a \\
    \end{tabular}
    \medskip
    \caption{More four-qubit stationary points (see \Cref{sec:overview}).}\label{tab:moreStatPts}
\end{table}

\subsubsection{Overview of stationary points}\label{sec:overview}
Given a point $\varphi\in\SS^{7}$, choose a chart $\alpha:U\to\RR^{7}$ on an open set $U$ containing $\varphi$ defined by \[\alpha^{-1}(x_2,\dots,x_8)=((1-\sum_{i=2}^8 x_i^2)^{1/2},x_2,\dots,x_8).\]
Given smooth $g:\SS^{7}\to\RR$, define $\textbf{H}(g)$ to be the Hessian matrix of the function $g\circ\alpha^{-1}:\RR^{7}\to\RR$.

In \Cref{tab:moreStatPts} we list the values attained on $|\mathcal{F}_1|$, $|\mathcal{F}_3|$, $|\mathcal{F}_4|$, and $|\mathcal{F}_6|$ for each of the points introduced in \Cref{sec:F3computation,sec:F4computation}. If $\varphi\in\SS^7$ is a nonvanishing stationary point of $|\mathcal{F}_i|$, we also record the 3-tuple $(\lambda^{-},\lambda^0,\lambda^+)$, where $\lambda^-$, $\lambda^0$, and $\lambda^+$ is the number of negative, zero, and positive eigenvalues respectively of $\textbf{H}(|\mathcal{F}_i|)$ evaluated at $\varphi$. We write n/a to indicate that a point is not stationary for a function. All Hessian matrices that we calculate have one zero eigenvalue; this is not surprising because $|\mathcal{F}_i|$ is fixed by the action of the one-dimensional real Lie group $\mathbb{T}$. An asterisk indicates that an item also appears in \Cref{table}.

\section{Sparse quantum error correcting codes}\label{sec:code}

\subsection{Uniform states and pure codes} A point $\ket{\varphi}\in (\CC^2)^{\otimes n}$ is called \textit{$r$-uniform} if every $r$-qubit reduced density matrix of $\ket{\varphi}\bra{\varphi}$ is proportional to the identity (we write $\bra{\varphi}$ for the conjugate transpose of the vector $\ket{\varphi}$). This extends the notion of critical points since a point is critical if and only if it is $1$-uniform. Note that $r\leq \lfloor n/2 \rfloor$ whenever $r$-uniform states exist in $(\CC^2)^{\otimes n}$. An $r$-uniform state such that $r=\lfloor n/2 \rfloor$ is called \textit{absolutely maximally entangled} (AME). 
These AME states have applications in quantum secret sharing \cite{PhysRevA.86.052335}, open-destination teleportation \cite{Helwig:2013qoq}, and quantum error correction \cites{Huber2020,681316}. There has been much work on the existence and constructions of AME states and, more generally, $r$-uniform states \cites{PhysRevLett.118.200502,PhysRevA.90.022316,ZangTianFeiZuo}. 
In \cite{Borras_2007}, Borras et al.~were able to obtain five-qubit and six-qubit AME states by maximizing various entanglement measures with a numerical search algorithm. Using the states found in \Cref{tab:moreStatPts}, we show more broadly that stationary points of entanglement measures can be used to obtain AME states. In particular, we obtain five-qubit and six-qubit AME states by reversing a construction of Rains (\Cref{prop:rains}). Note that four-qubit AME states do not exist \cite{HIGUCHI2000213}.

\begin{definition}
    Let $\mathcal{C}$ be a $K$-dimensional subspace of $(\CC^2)^{\otimes n}$. We say that $\mathcal{C}$ is a \textit{pure quantum error correcting code} (or simply \textit{pure code}) with parameters $((n,K,d))$ if every point in $\mathcal{C}$ is $(d-1)$-uniform.
\end{definition}

Note that, although the definition of a pure quantum error correcting code given above is not standard, it is known to be equivalent to the standard definition \cite[Observation 1]{Huber2020}. The parameter $d$ is the \textit{distance} of the code and represents the amount of error the code can correct: a code with distance $d\geq 2t+1$ can correct errors that affect up to $t$ subsystems.

Let $V=V_1\otimes V_2\otimes\dots\otimes V_n$ where $V_i\cong \CC^2$ for all $1\leq i\leq n$. If $S$ is a subset of $\{1,2,\dots,n\}$ and $A\in\text{End}(V)$, then $\text{Tr}_S(A)$ is the partial trace of the operator $A$ tracing out the subsystems corresponding to the indices in $S$.

\begin{proposition}[Rains]\label{prop:rains}
    Let $\{\ket{\varphi_i}\}$ be an orthogonal basis of a pure code with parameters $((n,K,d))$. The image of the operator
    $
    \text{Tr}_{\{1\}}(\ket{\varphi_1}\bra{\varphi_1}+\dots + \ket{\varphi_{K}}\bra{\varphi_{K}})
    $
    is a pure code with parameters $((n-1,2K,d-1))$.
\end{proposition}
\begin{proof}
    See \cite[Theorem 19]{Rains1996QuantumWE}.
\end{proof}

\begin{table}[t]
    \centering
    \begin{itemize}
\item[1.] The point $\varphi_6$ and a $\tilde{W}$-equivalent version of itself:
\begin{flalign*}
\quad\qquad\qquad\ket{\Phi_{0}}&= (\sqrt{2}+\im)\ket{v_1}+(\sqrt{2}-\im)\ket{v_2}+\im\ket{v_3}+\im\ket{v_4},& \\
\ket{\Phi_{1}}&= -\im\ket{v_1}+\im\ket{v_2}+(\sqrt{2}+\im)\ket{v_3}+(-\sqrt{2}+\im)\ket{v_4}.&
\end{flalign*}
\item[2.] The pair of points $\varphi_8\cong \ket{HS}$ and $\varphi_9\cong\ket{HD}$:
\begin{flalign*}
\quad\qquad\qquad\ket{\Phi_{0}}&=(\sqrt{3}+\im)\ket{v_1}+(\sqrt{3}-\im)\ket{v_2}+2\im\ket{v_3},& \\
\ket{\Phi_{1}}&=-\sqrt{2}e^{\pi \im/3}\ket{v_1}+\sqrt{2}e^{\pi \im/3}\ket{v_2}+\sqrt{2}e^{\pi \im/3}\ket{v_3}+\sqrt{2}(1+e^{-\pi \im/3})\ket{v_4}.&
\end{flalign*}
\item[3.] The point $\ket{\varphi_{10}}$ and a $\tilde{W}$-equivalent version of itself:
\begin{flalign*}
    \quad\qquad\qquad\ket{\Phi_{0}}=&\begin{multlined}[t]
    (2\sqrt{2}-\sqrt{3}+\im)\ket{v_1}+(\sqrt{3}-\im)\ket{v_2}\\+(\sqrt{3}-\sqrt{2}+\sqrt{6}\im-\im)\ket{v_3}+(\sqrt{3}-\sqrt{2}-\sqrt{6}\im+3\im)\ket{v_4},
    \end{multlined}&\\
    \ket{\Phi_{1}}=&\begin{multlined}[t]
    (\sqrt{2}-\sqrt{3}+\sqrt{6}\im-\im)\ket{v_1}+(\sqrt{3}-\sqrt{2}+\sqrt{6}\im-3\im)\ket{v_2}\\+(2\sqrt{2}-\sqrt{3}-\im)\ket{v_3}+(-\sqrt{3}-\im)\ket{v_4}.
    \end{multlined}&
\end{flalign*}
\item[4.] The cluster state $\psi_5\cong\ket{C_1}$ and a $\tilde{W}$-equivalent version of itself:
\begin{flalign*}
\quad\qquad\qquad\ket{\Phi_{0}}&=e^{\pi\im/4}\ket{v_1}+e^{-\pi\im/4}\ket{v_2},& \\
\ket{\Phi_{1}}&=e^{-\pi\im/4}\ket{v_3}-e^{\pi\im/4}\ket{v_4}.& 
\end{flalign*}
\item[5.] The point $\psi_6\cong\ket{BSSB}$ and a $\tilde{W}$-equivalent version of itself:
\begin{flalign*}
\quad\qquad\qquad\ket{\Phi_{0}}&=e^{\pi\im/4}\ket{v_1}+e^{-\pi\im/4}\ket{v_2}+\ket{v_3}-\im \ket{v_4},& \\
\ket{\Phi_{1}}&=e^{\pi\im/4}\ket{v_1}+e^{-\pi\im/4}\ket{v_2}-\ket{v_3}+\im \ket{v_4}.&
\end{flalign*}
\end{itemize}
    \caption{Pairs of four-qubit critical states that give rise to five- and six-qubit AME states.}
    \label{tab:AME_constructors}
\end{table}

\subsection{Pure codes from stationary points}\label{sec:sparselist} Given a pair of critical states $\ket{\Phi_{0}},\ket{\Phi_{1}}\in\mathfrak{a}$ in the Cartan subspace of $\mathcal{H}_4$, we can attempt to construct a six-qubit AME state by defining
\begin{equation*}
\ket{\Phi}=\ket{00}\otimes\ket{\Phi_0}+\ket{01}\otimes\ket{\Phi_1}-\ket{10}\otimes\overline{\ket{\Phi_1}}+\ket{11}\otimes\overline{\ket{\Phi_0}},
\end{equation*}
where $v\mapsto \overline{v}$ is entrywise complex conjugation. The expression above was derived by reversing the construction of Rains (\Cref{prop:rains}) and can be explained as follows. If $\ket{\Phi}$ is an AME state, then its span is a $((6,1,4))$ pure code. By \Cref{prop:rains}, the image of
$
\text{Tr}_{\{1\}}(\ket{\Phi}\bra{\Phi})
$
is a $((5,2,3))$ pure code. That is, every point in the image corresponds to a five-qubit AME state. A natural orthogonal basis of this code is given by the vectors
\[
\ket{0_L}=\ket{0}\otimes\ket{\Phi_0}+\ket{1}\otimes\ket{\Phi_1},\quad\text{and}\quad \ket{1_L}=-\ket{0}\otimes\overline{\ket{\Phi_1}}+\ket{1}\otimes\overline{\ket{\Phi_0}}.
\]
Now \Cref{prop:rains} applies again: the image of \[\text{Tr}_{\{1\}}(\ket{0_L}\bra{0_L}+\ket{1_L}\bra{1_L})=\text{Tr}_{\{1\}}(\text{Tr}_{\{1\}}(\ket{\Phi}\bra{\Phi}))=\text{Tr}_{\{1,2\}}(\ket{\Phi}\bra{\Phi})\] is the $((4,4,2))$ code $\mathfrak{a}$. Note that, due to the sparsity of points in $\mathfrak{a}$, any five-qubit or six-qubit AME
state obtained in this way must be sparse: at least half of the entries of the state vectors are zero in the computational basis. We find that it is possible to construct a six-qubit AME state---and thus a family of pure codes---by using points from \Cref{tab:moreStatPts}.
Specifically, a six-qubit AME state arises from the pairs of points $(\ket{\Phi_0},\ket{\Phi_1})$ in \Cref{tab:AME_constructors}.
Here, we use the basis
\[
\ket{v_1}=\ket{u_1}+\ket{u_2},\quad \ket{v_2}=\ket{u_1}-\ket{u_2},\quad \ket{v_3}=\ket{u_3}+\ket{u_4},\quad \ket{v_4}=\ket{u_3}-\ket{u_4}.
\]
As before, we do not normalize vectors.

Six four-qubit critical states appear in \Cref{tab:AME_constructors}. Four of them, $\varphi_8\cong\ket{HS}$, $\varphi_9\cong\ket{HD}$, $\psi_5\cong\ket{C_1}$, and $\psi_6\cong\ket{BSSB}$ can be found in prior literature (\Cref{sec:selection}). Two are new to us, $\varphi_6$ and $\varphi_{10}$. Interestingly, these six are precisely the states in \Cref{tab:moreStatPts} that vanish on $\mathcal{F}_1$. The states $\varphi_8$, $\varphi_9$, and $\varphi_{10}$ also vanish on $\mathcal{F}_4$. The cluster state $\psi_5$ also vanishes on $\mathcal{F}_3$. The BSSB state $\psi_6$ additionally vanishes on $\mathcal{F}_3$ and $\mathcal{F}_6$.

The six-qubit AME state and the $((5,2,3))$ pure code are known to be unique up to local unitary equivalence \cites{746807,Huber2020}. Therefore, we cannot say that the AME states or the pure codes constructed here are new. We only claim that the construction is new.

\section{Discussion and conclusions}\label{sec:conclusions}
In this paper, we study the entanglement of four-qubit systems by considering various entanglement measures on pure state vectors. Given an entanglement measure $E:\mathcal{H}_4\to [0,\infty)$ one may ask if states that are already known to be highly entangled and useful are also highly entangled with respect to $E$. If there is evidence that $E$ detects a useful notion of entanglement, one might ask to find all local maximizers and stationary points of $E$ in the hope that these special points have interesting entanglement properties. Our approach is in the same vein as preceding work \cites{OsterlohHD,Borras_2007,GourWallach2010,jaffali2023maximally,HIGUCHI2000213}. However, we also extend these techniques, arguing that it is meaningful to find all stationary points of entanglement measures and not just the local maximizers. One reason for this is the observation in \Cref{sec:selection} that many four-qubit critical states that have been studied in the past show up as stationary points of naturally chosen entanglement measures. The cluster state, which exhibits useful entanglement properties, is \textit{not} a local maximizer of any of our entanglement measures, but it is a stationary point. Our exploration also reveals that a state can be local maximizer of some entanglement measure while only being a stationary point for another.

In \Cref{sec:sparselist}, we give evidence for the efficacy of our philosophy: six of the stationary points we compute in this paper give rise to pure codes by reversing a construction of Rains (\Cref{prop:rains}). This is notable since these points were obtained by solving a seemingly distant calculus problem. Of these six, four appear in prior literature for different reasons, while two are new to the best of our knowledge. We note that these six are precisely the ones among those stationary points found that vanish on the SLOCC invariant polynomial $\mathcal{F}_1$. Given the various interesting non-generic properties these six critical states possess, further study may be fruitful. The general question of when Rains's construction can be reversed is likely still open (see also Section 7 of \cite{Huber2020}).

We choose to focus on the case of four qubits because of the connection with Vinberg theory. In addition, we choose entanglement measures that arise from SLOCC invariant polynomials. This allows us to simplify and make tractable the computation of the stationary points. The idea is that the four-qubit state space $\mathcal{H}_4$ embeds as the grade 1 piece of the $\mathbb{Z}_2$-graded Lie algebra $\mathfrak{so}_8$. Consequently, the space $\mathcal{H}_4$ behaves nicely. For example, the algebra of SLOCC invariant polynomials on $\mathcal{H}_4$ is equivalent to the algebra of $W$-invariant polynomials on the four-dimensional Cartan subspace $\mathfrak{a}\subset\mathcal{H}_4$ (\Cref{prop:restriction}), where $W$ is a finite group known as the Weyl group. This simplifies the computation of SLOCC invariant polynomials and has already been applied to quantum information in previous work \cites{jaffali2023maximally,GourWallach2010}. Our main mathematical contribution is \Cref{prop:crit}, which in conjunction with \Cref{prop:critorbit} and \Cref{prop:restriction} allows us also to simplify the calculation of stationary points of SLOCC invariant polynomials. This deepens the connection between Vinberg theory and the Kempf-Ness theorem studied by Wallach \cite{WallachGIT}.

\section{Acknowledgments} We thank Jonathan Hauenstein for helping us carry out a computation with Bertini.

\newcommand{\arxiv}[1]{\href{http://arxiv.org/abs/#1}{{\tt arXiv:#1}}}
\bibliographystyle{amsplain}
\bibliography{ref}

\end{document}